\documentclass[11pt,reqno]{amsart}
\usepackage[T1]{fontenc}
\usepackage{lmodern}
\usepackage[english]{babel}
\usepackage{amsmath,amssymb,amsthm,mathtools}
\usepackage[margin=1.2in]{geometry}
\usepackage[colorlinks=true,linkcolor=blue,citecolor=blue,urlcolor=blue]{hyperref}
\hypersetup{
  pdftitle={Asymptotically unit-rate storage codes from binary BCH codes},
  pdfauthor={Aryeh Lev Zabokritskiy (Yohananov)},
  pdfkeywords={Storage codes, BCH codes, Cayley graphs, polynomial method},
  pdflang={en},
  pdfdisplaydoctitle=true
}
\newtheorem{theorem}{Theorem}[section]

\newtheorem{proposition}[theorem]{Proposition}
\newtheorem{corollary}[theorem]{Corollary}
\theoremstyle{remark}

\DeclareMathOperator{\rank}{rank}
\DeclareMathOperator{\spanF}{span}
\DeclareMathOperator{\wt}{wt}
\newcommand{\F}{\mathbb F}
\newcommand{\Cay}{\operatorname{Cay}}
\title{Asymptotically unit-rate storage codes from binary BCH codes}
\author{Aryeh Lev Zabokritskiy (Yohananov)}
\address{Department of Computer Science, Tel-Hai University of Kiryat Shmona in the Galilee, Israel; and MIGAL---Galilee Research Institute, Israel}
\email{yuhanalev@telhai.ac.il}
\date{}
\begin{document}

\begin{abstract}
A storage code on a graph assigns a symbol to each vertex so that the symbol can be recovered from its neighbors. We prove that the binary full-parity storage codes on triangle-free coset graphs of primitive BCH codes have rate tending to one for every fixed error-correction capability of at least two. The conclusion also holds for an increasing error budget under an explicit growth condition. The proof bounds the binary rank of convolution matrices associated with polynomial maps. Low coordinate degree forces cancellation in the polynomial representing the image, and a polynomial rank bound converts this cancellation into a quantitative estimate for the storage rate. For the double-error-correcting family over $\F_{2^m}$, we prove that the storage deficiency is $\Theta(((1+\sqrt5)/4)^m)$. This refinement follows from an exact Fibonacci formula for a polynomial coefficient rank and a matching exponential lower bound after finite-field evaluation.
\end{abstract}
\maketitle

\section{Introduction}

A storage code on a graph is a collection of assignments of symbols to the vertices in which each symbol is determined by the symbols at neighboring vertices. The graph prescribes the available recovery sets, and the objective is to maximize the rate of the code. Complete graphs support rates tending to one through a single parity check. Requiring the graph to be triangle-free removes this construction and leads to a substantially different question.

Barg and Z\'emor~\cite{BZ} studied binary full-parity storage codes on coset graphs of linear codes and constructed families with rates above one half. They also identified BCH and Reed--Muller codes as natural sources of further examples. Constructions yielding asymptotically unit-rate binary storage codes on triangle-free graphs were obtained by Golovnev and Haviv~\cite{GH}, Barg, Schwartz, and Yohananov~\cite{BSY}, and Huang and Xiang~\cite{HX}. Deng, Huang, Weng, and Xiang~\cite{DHWX} proved that the family arising from the double-error-correcting BCH code has asymptotically unit rate, and studied further constructions with two coordinates over an extension field.

Here we consider the primitive binary BCH code with designed distance $2s+1$, whose parity-check columns can be written as
\[
 (a,a^3,\ldots,a^{2s-1}),\qquad a\in\F_{2^m}^{*}.
\]
We prove that its coset graph supports a full-parity storage code of rate tending to one for every fixed $s\ge2$. More quantitatively, writing $R_{m,s}$ for this storage rate and $d_s=1+\lfloor\log_2s\rfloor$, we obtain
\[
 1-R_{m,s}\le 2\exp\!\left(-\frac{m}{2s d_s^2}\right).
\]
The same estimate proves convergence to one whenever $s=s(m)\ge2$ and $s(\log(s+1))^2=o(m)$. The underlying graphs are connected and triangle-free. The parameter $s$ refers to the error-correction capability of the BCH code used to define the graph.

The argument uses the Boolean algebraic degree of the parity-check columns. We establish a rank bound for convolution by the image of a polynomial map over $\F_2$ and apply it to the BCH power map. The key step is that low coordinate degrees force high-degree coefficients to vanish in the polynomial recording the parity of the fibres. The subsequent rank step is the polynomial rank argument of Croot, Lev, and Pach~\cite{CLP}. The relevant dimensions explain the scope of the application: for fixed $s$, the BCH syndrome space has dimension at most $sm$, whereas the nonconstant quadratic evaluation map has target dimension $m(m+1)/2$.

For the double-error-correcting family we also determine the exponential order of the rank. With $\varphi=(1+\sqrt5)/2$, its storage deficiency is $\Theta((\varphi/2)^m)$. The proof computes an exact Fibonacci formula for the coefficient rank of the polynomial used in~\cite{DHWX}, and derives matching exponential upper and lower bounds after evaluation over $\F_{2^m}$. This sharpens the displayed rank estimate in that work, whose exponential base is $2+\sqrt2$ before normalization.

A related open problem comes from the quadratic Reed--Muller construction proposed by Barg and Z\'emor~\cite[author version, Section~7.2, item~3]{BZ}, who left the determination of its storage-code rates open. We do not settle whether these rates tend to one. Instead, we identify the corresponding Cayley graph with the rank-one graph on symmetric binary matrices and express the unit-rate question as a binary matrix-rank problem in that setting. We also give an equivalent formulation through the intersection of the graph's cut and cycle spaces. These connections offer further approaches to the Reed--Muller problem: rank and dimension estimates in these settings give storage-rate bounds, while storage-code estimates constrain the same algebraic and graph-theoretic objects.

We first describe the storage construction and its convolution matrix. We then prove the polynomial rank bound and apply it to the BCH family. The double-error-correcting case receives a sharper analysis through its coefficient spaces. The final section gives these reformulations for the quadratic Reed--Muller family.

\section{Storage codes and polynomial maps}\label{sec:storage}

The storage rate in the construction is determined by a binary matrix indexed by the vertices of a Cayley graph. We describe this matrix explicitly so that its rank can be estimated without computing the dimension of the original BCH code.

Let $G$ be a finite simple graph with vertex set $V$ and adjacency matrix $A(G)$. The binary full-parity storage code on $G$ is
\begin{equation}\label{eq:storage-code}
 C(G)=\ker_{\F_2}(I+A(G)).
\end{equation}
For every $c\in C(G)$ and $v\in V$, the parity check at $v$ gives
\[
 c_v=\sum_{u\in N_G(v)}c_u.
\]
Thus each coordinate can be recovered from its neighbors. The rate of this particular code is
\begin{equation}\label{eq:rate}
 R(C(G))=1-\frac{\rank_{\F_2}(I+A(G))}{|V|}.
\end{equation}
This is the full-parity construction of~\cite[author version, Section~3]{BZ}; the identity matrix supplies the symbol being recovered.

For a binary vector space $V$ and a set $S\subseteq V\setminus\{0\}$, the Cayley graph $\Cay(V,S)$ has vertex set $V$, with distinct vertices $x,y$ adjacent when $x+y\in S$. It is connected precisely when $S$ spans $V$. If $H$ is a binary matrix whose distinct nonzero columns are the elements of $S$, the syndrome map identifies this graph with the coset graph of $\ker H$, using $\operatorname{im}H$ as the vertex space.

We will use connection sets parametrized by a map $h:\F_2^m\to\F_2^D$. All sums in the following definition are in $\F_2$:
\begin{equation}\label{eq:matrix}
 M_h(x,y)=\sum_{a\in\F_2^m}\mathbf1\{x+y=h(a)\},
 \qquad x,y\in\F_2^D.
\end{equation}
When $h$ is injective and $h(0)=0$, this is exactly $I+A(G)$ for $G=\Cay(\F_2^D,h(\F_2^m)\setminus\{0\})$. For a general map, the entry records the parity of the number of preimages of $x+y$.

Every function $\F_2^m\to\F_2$ has a unique multilinear polynomial representation. Its Boolean algebraic degree is the degree of this representation. A coordinate degree bound for $h$ is preserved by a binary linear change of coordinates in the target. In particular, one may replace a redundant target space by the span of the image: every linear coordinate on that span extends to a linear functional on the original target.

Retaining a larger ambient space gives one identical connected component on each coset of the span. Both the rank of the full-parity matrix and its order are multiplied by the number of cosets. Thus its normalized rank, and hence the storage rate, agree with those of a single component.

\section{A rank bound for polynomial maps}\label{sec:rank}

We now bound the rank of \eqref{eq:matrix} in terms of the input dimension and coordinate degree of $h$. The proof first represents the parity of each fibre by a polynomial of low degree. Splitting its monomials between the two matrix indices then bounds the rank.

For an integer $t$, put
\[
 B_D(t)=\sum_{j=0}^{\min(t,D)}\binom Dj,
\]
with $B_D(t)=0$ when $t<0$. Thus $B_D(t)$ counts the squarefree monomials of degree at most $t$ in $D$ variables.

\begin{theorem}\label{thm:polynomial-rank}
Let $m,D,d\ge1$ be integers, and let $h:\F_2^m\to\F_2^D$ have Boolean coordinate degrees at most $d$. Set $k=\lceil m/d\rceil$. If $k>D$, then $M_h=0$. If $k\le D$, then
\begin{equation}\label{eq:exact-rank}
 \rank_{\F_2}M_h
 \le B_D\!\left(\left\lfloor\frac{D-k}{2}\right\rfloor\right)
   +B_D\!\left(\left\lceil\frac{D-k}{2}\right\rceil-1\right).
\end{equation}
In either case,
\begin{equation}\label{eq:exp-rank}
 \frac{\rank_{\F_2}M_h}{2^D}
 \le \min\left\{\frac12,\;2\exp\!\left(-\frac{k^2}{2D}\right)\right\}.
\end{equation}
\end{theorem}

\begin{proof}
Write $h=(h_1,\ldots,h_D)$ and define the multilinear polynomial
\begin{equation}\label{eq:fibre-polynomial}
 F(z)=\sum_{a\in\F_2^m}\prod_{i=1}^D(1+z_i+h_i(a)).
\end{equation}
For a fixed $a$, the product is one precisely when $z=h(a)$. Hence $M_h(x,y)=F(x+y)$. For $S\subseteq\{1,\ldots,D\}$, the coefficient of $z_S=\prod_{i\in S}z_i$ in $F$ is
\begin{equation}\label{eq:coefficient}
 [z_S]F=\sum_{a\in\F_2^m}\prod_{i\notin S}(1+h_i(a)).
\end{equation}
The summand, viewed as a Boolean polynomial in $a$, has degree at most $d(D-|S|)$. A Boolean polynomial of degree less than $m$ sums to zero over $\F_2^m$: each squarefree monomial of degree $j<m$ takes the value one at $2^{m-j}$ points. It follows that \eqref{eq:coefficient} vanishes whenever $D-|S|<k$. Consequently $\deg F\le D-k$, and $F=0$ if $k>D$.

Suppose $k\le D$, and put
\[
 u=\left\lfloor\frac{D-k}{2}\right\rfloor,
 \qquad v=\left\lceil\frac{D-k}{2}\right\rceil-1.
\]
Every monomial $x_Ay_B$ in $F(x+y)$ satisfies $|A|+|B|\le D-k$. Group those with $|A|\le u$ by their $x$-monomial. Every remaining monomial has $|B|\le v$, so these can be grouped by their $y$-monomial. Thus
\[
 F(x+y)=\sum_{|A|\le u}x_Ap_A(y)
       +\sum_{|B|\le v}q_B(x)y_B
\]
for suitable polynomials $p_A,q_B$. Each summand, evaluated on all pairs $(x,y)$, gives a matrix of rank at most one. Rank subadditivity proves \eqref{eq:exact-rank}. This is the polynomial separation used in the proof of~\cite[Lemma~1]{CLP}.

For the exponential estimate, let $X$ have binomial distribution with parameters $D$ and $1/2$. Since $v\le u\le(D-k)/2$, the right side of \eqref{eq:exact-rank}, divided by $2^D$, is at most $2\Pr(X-D/2\le-k/2)$. For $t>0$, the exponential Markov inequality gives
\[
 \Pr(X-D/2\le-k/2)
 \le e^{-tk/2}\cosh(t/2)^D
 \le \exp(-tk/2+Dt^2/8).
\]
Here $\log\cosh w\le w^2/2$, which follows by integrating $\tanh w\le w$ for $w\ge0$. Taking $t=2k/D$ gives the exponential term in \eqref{eq:exp-rank}.

Finally, for $b\in\F_2^D$, let $T_b$ be the translation matrix acting on functions by $(T_bf)(x)=f(x+b)$. Then $M_h=\sum_aT_{h(a)}$. These matrices commute and $T_b^2=I$, so in characteristic two
\[
 M_h^2=\sum_{a\in\F_2^m}T_{h(a)}^2=2^m I=0.
\]
Therefore $\operatorname{im}M_h\subseteq\ker M_h$ and $\rank M_h\le2^{D-1}$.
\end{proof}

The exact rounding in \eqref{eq:exact-rank} is useful at the endpoints. When $k=D$, it gives rank at most one. This is attained by $h(a)=a$ with $D=m$ and $d=1$, for which $M_h$ is the all-one matrix.

\begin{corollary}\label{cor:criterion}
For a sequence of maps as in Theorem~\ref{thm:polynomial-rank}, the normalized rank $2^{-D}\rank M_h$ tends to zero if
\[
 \frac{m^2}{d^2D}\longrightarrow\infty.
\]
If the maps are injective and send zero to zero, the corresponding full-parity storage rates tend to one.
\end{corollary}

\begin{proof}
The hypothesis implies $k^2/D\to\infty$. Apply \eqref{eq:exp-rank}, followed by \eqref{eq:rate} in the injective case.
\end{proof}

\section{BCH coset graphs}\label{sec:bch}

For BCH parity checks, the number of binary output coordinates grows linearly with the extension degree when the error budget is fixed. This allows the rank estimate to give asymptotically unit storage rate.

Let $m\ge1$ and $s\ge2$ be integers, put $q=2^m$, and suppose that $2s+1\le q-1$. Define
\begin{equation}\label{eq:bch-map}
 h_s:\F_q\longrightarrow\F_q^s,
 \qquad h_s(a)=(a,a^3,\ldots,a^{2s-1}).
\end{equation}
After choosing a binary basis of $\F_q$, use the columns $h_s(a)$ for $a\ne0$ as a binary parity-check matrix $H_{m,s}$. Its kernel is the primitive narrow-sense binary BCH code $B_{m,s}$ of designed distance $2s+1$. Indeed the odd power checks imply all power checks with exponents $1,\ldots,2s$, since an even power sum is the square of the power sum with half its exponent.

For completeness, these checks exclude every nonzero codeword of weight at most $2s$. If such a word had support $a_1,\ldots,a_w$, then
\[
 \sum_{i=1}^w a_i^j=0,\qquad 1\le j\le w.
\]
The matrix $(a_i^j)_{1\le j,i\le w}$ is invertible: its determinant is the product of the nonzero $a_i$ and the Vandermonde determinant on the distinct $a_i$. This contradicts its annihilating the all-one vector. Thus $B_{m,s}$ corrects at least $s$ errors.

Let
\[
 V_{m,s}=\spanF_{\F_2}\{h_s(a):a\in\F_q\},
 \qquad D_{m,s}=\dim_{\F_2}V_{m,s},
\]
and define
\begin{equation}\label{eq:bch-graph}
 G_{m,s}=\Cay(V_{m,s},h_s(\F_q^*)),
 \qquad C_{m,s}=C(G_{m,s}).
\end{equation}
The graph is connected by its definition, and $D_{m,s}\le sm$. Its degree equals the length $q-1$ of the seed code $B_{m,s}$, whereas the storage code has length $2^{D_{m,s}}$.

\begin{theorem}\label{thm:bch}
For $s\ge2$ and $2s+1\le2^m-1$, the graph $G_{m,s}$ is triangle-free and
\begin{equation}\label{eq:bch-rate}
 1-R(C_{m,s})\le
 2\exp\!\left(-\frac{m}{2s(1+\lfloor\log_2s\rfloor)^2}\right).
\end{equation}
Consequently $R(C_{m,s})\to1$ as $m\to\infty$ for every fixed $s\ge2$. The same conclusion holds for $s=s(m)\ge2$ satisfying $s(\log(s+1))^2=o(m)$.
\end{theorem}

\begin{proof}
The first coordinate in \eqref{eq:bch-map} makes $h_s$ injective, and $h_s(0)=0$. A triangle in $G_{m,s}$ would give distinct nonzero $a,b,c$ with $h_s(a)+h_s(b)+h_s(c)=0$. The first two coordinates imply $c=a+b$ and
\[
 0=a^3+b^3+(a+b)^3=ab(a+b),
\]
which is impossible.

For a nonnegative integer $e$, write $\wt_2(e)$ for the number of ones in its binary expansion. The map $a\mapsto a^e$ is a product of $\wt_2(e)$ Frobenius maps. Each Frobenius map is $\F_2$-linear, and field multiplication is bilinear over $\F_2$. Thus its binary coordinate degrees are at most $\wt_2(e)$. It follows that the coordinate degrees of $h_s$ are at most
\[
 d_s=\max_{1\le j\le s}\wt_2(2j-1)
     =1+\lfloor\log_2s\rfloor.
\]
To verify the equality, write $2^r\le s<2^{r+1}$. The exponent $2^{r+1}-1$ occurs and has $r+1$ ones; every integer with at least $r+2$ ones is at least $2^{r+2}-1>2s-1$.

Pass to binary coordinates on $V_{m,s}$. The coordinate degree bound is unchanged. Equations \eqref{eq:matrix} and \eqref{eq:bch-graph} identify $M_{h_s}$ with $I+A(G_{m,s})$, and Theorem~\ref{thm:polynomial-rank} gives
\[
 1-R(C_{m,s})
 \le2\exp\!\left(-\frac{\lceil m/d_s\rceil^2}{2D_{m,s}}\right)
 \le2\exp\!\left(-\frac{m}{2s d_s^2}\right).
\]
This proves \eqref{eq:bch-rate}. For fixed $s$ the exponent tends to $-\infty$. Under the stated growth condition, $s d_s^2=o(m)$, so the same argument applies. That condition also gives $s=o(m)$, ensuring $2s+1\le2^m-1$ for all sufficiently large $m$.
\end{proof}

For $s=2$, the construction is the BCH family treated in~\cite{DHWX}. For $s=3$, the map is $(a,a^3,a^5)$, its coordinate degrees are at most two, and \eqref{eq:bch-rate} gives
\[
 1-R(C_{m,3})\le2e^{-m/24}.
\]
The argument applies uniformly to all the odd power checks in \eqref{eq:bch-map}; it uses an upper bound on the rank of the storage matrix rather than an exact rank formula.

\section{Fibonacci ranks for the double-error-correcting family}\label{sec:fibonacci}

The general degree estimate proves convergence of the rate for every fixed error budget. For two errors, a more detailed analysis determines its exponential order. We use the polynomial representation from~\cite[Section~II]{DHWX}, and compute its coefficient rank by tracking an intersection of coefficient spaces.

For $m\ge1$, put $q=2^m$ and let $\widehat G_m=\Cay(\F_q^2,\{(a,a^3):a\in\F_q^*\})$, with adjacency matrix $A_m$. For the admissible two-error parameters, its connected components are copies of the graph $G_{m,2}$ defined in \eqref{eq:bch-graph}.

The rank of $I+A_m$ also occurs in the theory of almost perfect nonlinear (APN) functions. For $f:\F_q\to\F_q$, put $\Gamma_f=\{(a,f(a)):a\in\F_q\}$. The binary incidence rank of the points of $\F_q^2$ and the translates of $\Gamma_f$ is called the $\Gamma$-rank of $f$; see~\cite[Section~2]{EP} and~\cite[Section~3.3]{Kaleyski}. Indexing the blocks by all translations, the incidence entry for point $(x,b)$ and block $\Gamma_f+(y,c)$ is one precisely when $b+c=f(x+y)$. For the Gold function $f(a)=a^3$, this is exactly the matrix $I+A_m$. Repeated blocks, if any, do not change its rank. Thus the evaluated-rank bounds below also give bounds on this established invariant of $a\mapsto a^3$.

Small-dimensional values of the $\Gamma$-rank of $a\mapsto a^3$ were tabulated by Edel and Pott~\cite[Tables~3--4]{EP}. The bounds below determine its exponential order as the extension degree tends to infinity.

As ordinary polynomials, also allowing $m=0$, set
\[
 d(x,u;y,v)=x^2y+xy^2+u+v,\qquad P_m=d^{\,2^m-1}\quad(m\ge0).
\]
For $m\ge1$ and vertices $(x,b),(y,c)$ of $\widehat G_m$, membership of their difference in $\{(a,a^3):a\in\F_q\}$ is the equation $b+c+(x+y)^3=0$. The bijective row and column substitutions $u=b+x^3$, $v=c+y^3$ change its left side to $d(x,u;y,v)$. Thus $I+A_m$ is permutation equivalent to the evaluation of $1+P_m$.

We first compute the ordinary coefficient rank. Let $K_m$ be the coefficient matrix of $P_m$, with rows indexed by monomials in $(x,u)$ and columns by monomials in $(y,v)$, and put $r_m=\rank_{\F_2}K_m$, as in~\cite[Definition~4]{DHWX}. These are ordinary polynomial coefficients, before reduction by finite-field identities. Write $f_0=0$, $f_1=1$, and $f_{j+2}=f_{j+1}+f_j$ for the Fibonacci numbers.

\begin{theorem}\label{thm:fibonacci-coefficients}
For every $m\ge0$,
\begin{equation}\label{eq:fibonacci-rank}
 r_m=2^m f_{m+2}.
\end{equation}
Equivalently, $r_0=1$, $r_1=4$, and $r_{m+1}=2r_m+4r_{m-1}$ for $m\ge1$.
\end{theorem}

\begin{proof}
We regard coefficient vectors as elements of $E=\F_2[x,u]$, with its monomial basis, and identify the column variables $(y,v)$ with $(x,u)$. Each $K_m$ has finite support and is symmetric. Let $X$ be multiplication by $x$ on $E$; its transpose in the monomial basis deletes one $x$ when the exponent is positive and otherwise gives zero. Products involving $K_m$ have finite support, so their ranks can be computed after finite zero padding. For a subspace $W\subseteq E$, write $W^{[2]}=\{w^2:w\in W\}$.

Fix $m\ge1$ and abbreviate $K=K_{m-1}$, $U=\operatorname{im}K$, and $W=U+xU$. The identity $P_m=dP_{m-1}^2$ splits the coefficient matrix according to the parities of the exponents. In the order $(0,0),(0,1),(1,0)$ for both pairs of variables, deleting the zero $(1,1)$ rows and columns and dividing the even parts of exponents by two gives
\begin{equation}\label{eq:parity-block}
 K_m=\begin{pmatrix}
 0&K&XK\\
 K&0&0\\
 KX^{\mathsf T}&0&0
 \end{pmatrix}.
\end{equation}
The two nonzero rectangular blocks are transposes and have disjoint row and column supports. Hence
\begin{equation}\label{eq:rank-sum}
 r_m=2\rank[\,K\ \ XK\,]=2\dim W.
\end{equation}

To obtain a recurrence, we determine how much the new image intersects its translate by $x$. Define
\[
 L=\{(Kc,KX^{\mathsf T}c):c\in E\}\subseteq U\oplus U,
 \qquad Z=\{g:(0,g)\in L\}.
\]
Transposing $[\,K\ \ XK\,]$ shows that $\dim L=\dim W$. The projection of $L$ onto its first coordinate is onto $U$, and therefore
\begin{equation}\label{eq:z-dimension}
 \dim Z=\dim W-\dim U.
\end{equation}
Returning from the parity blocks to polynomial coordinates, \eqref{eq:parity-block} identifies the image $U_m=\operatorname{im}K_m$ as
\begin{equation}\label{eq:image-description}
 U_m=W^{[2]}\oplus\{u f^2+xg^2:(f,g)\in L\}.
\end{equation}
The two summands have different exponent parities, as do the two displayed coordinates inside the second summand.

The key step behind the Fibonacci recurrence is the following intersection identity:
\begin{equation}\label{eq:intersection}
 U_m\cap xU_m=(xZ)^{[2]}\oplus xZ^{[2]}.
\end{equation}
In an equality between an element of \eqref{eq:image-description} and $x$ times another such element, the terms with odd $u$-exponent have opposite $x$-parities. Their $f$-coordinates must therefore both vanish. The remaining spaces on the two sides are
\[
 W^{[2]}\oplus xZ^{[2]},\qquad
 xW^{[2]}\oplus x^2Z^{[2]}.
\]
Their even-$x$ intersection is $x^2Z^{[2]}=(xZ)^{[2]}$, since $xZ\subseteq W$, and their odd-$x$ intersection is $xZ^{[2]}$, since $Z\subseteq W$. Both inclusions follow from $Z\subseteq U$. This proves \eqref{eq:intersection}, including that every vector in its right side belongs to both spaces.

Squaring is injective and $\F_2$-linear. Equations \eqref{eq:rank-sum}--\eqref{eq:intersection} consequently give
\[
 \dim(U_m\cap xU_m)=2\dim Z=r_m-2r_{m-1}.
\]
Applying \eqref{eq:rank-sum} at the next index yields
\[
 r_{m+1}=2\dim(U_m+xU_m)
 =4r_m-2(r_m-2r_{m-1})=2r_m+4r_{m-1}.
\]
Finally $P_0=1$ gives $r_0=1$, and the four terms of $d$ give $r_1=4$. The Fibonacci recurrence proves \eqref{eq:fibonacci-rank}.
\end{proof}

Evaluation can identify different monomials, so the rank of the storage matrix requires a separate estimate.

\begin{theorem}\label{thm:fibonacci-evaluation}
For every $m\ge1$,
\begin{equation}\label{eq:two-sided}
 2^m f_{m+1}\le\rank_{\F_2}(I+A_m)\le2^m f_{m+2}.
\end{equation}
Consequently, with $\varphi=(1+\sqrt5)/2$,
\[
 \rank_{\F_2}(I+A_m)=\Theta((1+\sqrt5)^m),
 \qquad 1-R(C_{m,2})=\Theta((\varphi/2)^m).
\]
The storage-rate assertion is taken as $m\to\infty$ in the admissible BCH range.
\end{theorem}

\begin{proof}
By the preceding identification, it remains to bound the rank of the evaluated matrix of $1+P_m$. Write
\[
 E_m=\bigl(P_m(x,u;y,v)\bigr)_{(x,u),(y,v)\in\F_q^2}
\]
for the binary evaluation matrix of $P_m$. Let $J$ sum rows over $u\in\F_q$ with $x$ fixed and repeat that sum at each value of $u$. Then $J^2=qJ=0$, so $I+J$ is invertible over $\F_2$. For every fixed $x,y,v$, exactly $q-1$ choices of $u$ make $d\ne0$, and hence the sum over $u$ of $P_m=d^{q-1}$ is one. Thus
\[
 (I+J)E_m=E_m+\mathbf1\mathbf1^{\mathsf T},
\]
where $\mathbf1$ is the all-one vector indexed by $\F_q^2$. The right side is the evaluation matrix of $1+P_m$, so it has the same binary rank as $E_m$. Evaluation factors through $K_m$, and extension of scalars preserves the rank of a binary matrix. This proves the upper bound in \eqref{eq:two-sided} from Theorem~\ref{thm:fibonacci-coefficients}.

For the lower bound, set $Q=2^{m-1}$, so that $q=2Q$. Frobenius gives
\[
 P_m=P_{m-1}(x^q y^Q+x^Q y^q+u^Q+v^Q).
\]
As a function on $\F_q^4$, this is represented by the reduction of
\begin{equation}\label{eq:high-bit}
 P_{m-1}(xy^Q+x^Qy+u^Q+v^Q)
\end{equation}
modulo $x^q=x$ and $y^q=y$. The individual $x,y$ degrees of $P_{m-1}$ are at most $2(Q-1)=q-2$, and its $u,v$ degrees are at most $Q-1$. Thus $P_{m-1}$ itself requires no reduction over $\F_q$, and the $u$- and $v$-degrees in \eqref{eq:high-bit} stay below $q$.

Use the reduced row basis $x^iu^j$ and column basis $y^iv^j$, with $0\le i,j<q$. Partition rows according to $j<Q$ or $j\ge Q$, and do the same for the $v$-exponent in the columns. The cross terms $xy^Q+x^Qy$ affect only the upper-left block, even after reduction. The terms $u^Q$ and $v^Q$ give the two off-diagonal blocks, each an embedded copy $K$ of the ordinary matrix $K_{m-1}$. The lower-right block is zero. The reduced coefficient matrix therefore has the form
\[
 \begin{pmatrix}B&K\\ K&0\end{pmatrix}.
\]
Its rank is at least $2\rank K$: selecting a nonsingular minor of $K$ in both off-diagonal positions gives a nonsingular minor of twice that order, independently of $B$. Theorem~\ref{thm:fibonacci-coefficients} gives $2\rank K=2^m f_{m+1}$.

Evaluation of all monomials of degree below $q$ at the $q$ distinct field elements is invertible, and the two-variable evaluation matrices are tensor products of these Vandermonde matrices. Thus reduced coefficient rank equals evaluated rank over $\F_q$; this is the evaluation principle in~\cite[Lemma~24]{DHWX}. The evaluated entries of $P_m$ are binary, so their rank is unchanged on extending $\F_2$ to $\F_q$. This proves the lower bound, including $m=1$, where $P_0=1$. The asymptotic statements now follow from $f_{m+j}=\Theta(\varphi^m)$ for fixed $j$, division by $q^2$, and the component normalization described in Section~\ref{sec:storage}.
\end{proof}

The estimate in~\cite[Theorems~2 and~9]{DHWX} uses the number of supported row monomials, whose exponential growth is $(2+\sqrt2)^m$. Theorem~\ref{thm:fibonacci-coefficients} instead computes the rank of the corresponding coefficient matrix, with growth $(1+\sqrt5)^m$. The lower bound in Theorem~\ref{thm:fibonacci-evaluation} shows that this latter base is also the exact exponential order of the evaluated storage-matrix rank. In terms of the ambient graph order $N=q^2$, the deficiency is $\Theta(N^{-\varepsilon/2})$, where $\varepsilon=\log_2(2/\varphi)=0.3057580863\ldots$.

\section{Quadratic evaluation and symmetric matrices}\label{sec:quadratic}

We now leave the BCH family and turn to the quadratic Reed--Muller construction proposed by Barg and Z\'emor~\cite[author version, Section~7.2, item~3]{BZ}. They left the determination of the associated storage-code rates open. The BCH argument does not establish convergence to one for this family, whose target dimension is quadratic rather than linear in the input dimension. Our purpose here is to express the remaining question in terms of symmetric binary matrices and, subsequently, the graph's cut and cycle spaces, so that information in either setting can be transferred to the storage problem.

The coordinates of the quadratic evaluation map are the nonconstant squarefree monomials of degree at most two, namely $x_i$ and $x_ix_j$ with $i<j$. These are the quadratic Reed--Muller evaluation coordinates after removing the constant coordinate. They can be assembled into a single symmetric matrix, which gives the desired reformulation.

For $m\ge1$, let $X_m=\operatorname{Sym}_m(\F_2)$, of dimension $D=m(m+1)/2$, and define $\nu_m:\F_2^m\to X_m$ by $\nu_m(x)=xx^{\mathsf T}$. Let $G_m^{\mathrm{sym}}$ have vertex set $X_m$, with distinct $S,T$ adjacent when $\rank(S+T)=1$. This rank-one relation belongs to the translation association scheme of symmetric bilinear forms; see~\cite[Section~2]{Schmidt}. All ranks of matrices with entries in $\F_2$ are taken over $\F_2$.

\begin{proposition}\label{prop:symmetric}
The graph $G_m^{\mathrm{sym}}$ is the connected Cayley graph generated by $\nu_m(\F_2^m\setminus\{0\})$. It has degree $2^m-1$, is triangle-free, and its full-parity matrix is
\begin{equation}\label{eq:symmetric-matrix}
 M_{\nu_m}(S,T)=\mathbf1\{\rank(S+T)\le1\}
              =\bigl(I+A(G_m^{\mathrm{sym}})\bigr)_{S,T}.
\end{equation}
\end{proposition}

\begin{proof}
The diagonal of $xx^{\mathsf T}$ is $x$, and its strict upper-triangular entries are $x_ix_j$. Thus $\nu_m$ is the quadratic evaluation map under the indicated linear identification of coordinates, and it is injective. Every nonzero symmetric matrix of rank one over $\F_2$ has the form $xx^{\mathsf T}$ for a unique nonzero $x$: a rank-one matrix is $uv^{\mathsf T}$, and symmetry forces the nonzero vectors $u,v$ to be proportional, hence equal over $\F_2$. This proves \eqref{eq:symmetric-matrix} and the degree assertion.

The generators span $X_m$. Indeed $\nu_m(e_i)=E_{ii}$, and
\[
 \nu_m(e_i+e_j)+\nu_m(e_i)+\nu_m(e_j)=E_{ij}+E_{ji}.
\]
For triangle-freeness, three distinct nonzero generators summing to zero would give $\nu_m(x)+\nu_m(y)=\nu_m(z)$, with $z=x+y$ by the diagonal coordinates. The remaining equation is $xy^{\mathsf T}+yx^{\mathsf T}=0$. If $x,y$ are distinct nonzero binary vectors they are linearly independent; a linear functional taking the values one on $x$ and zero on $y$, applied on the left, would give $y^{\mathsf T}=0$, a contradiction.
\end{proof}

Theorem~\ref{thm:polynomial-rank} applies with $d=2$ and $k=\lceil m/2\rceil$. Here $k/\sqrt D\to1/\sqrt2$, and the exact binomial bound gives
\begin{equation}\label{eq:rm-constant}
 \limsup_{m\to\infty}2^{-D}\rank M_{\nu_m}
 \le2\Phi(-1/\sqrt2)=0.4795001222\ldots,
\end{equation}
where $\Phi$ is the standard normal distribution function. To see this, standardize a binomial random variable with mean $D/2$ and standard deviation $\sqrt D/2$ in each of the two terms of \eqref{eq:exact-rank}. Both thresholds tend to $-1/\sqrt2$, so the central limit theorem gives \eqref{eq:rm-constant}. Thus this method yields a storage rate of at least $0.5204998778\ldots-o(1)$ for the quadratic family.

\emph{Open question.} Does the full-parity storage rate $R(C(G_m^{\mathrm{sym}}))$ tend to one as $m\to\infty$? Equivalently, does the binary rank satisfy
\begin{equation}\label{eq:rm-open}
 \rank_{\F_2}\!\left(\mathbf1\{\rank(S+T)\le1\}\right)_{S,T\in X_m}
 =o\!\left(2^{m(m+1)/2}\right).
\end{equation}
The two ranks in this expression act on different spaces: $S+T$ is an $m\times m$ matrix, while the matrix of indicators is indexed by all symmetric matrices. Formula \eqref{eq:rm-open} concerns linear dependence among these indicator rows. Counts of symmetric matrices by rank describe the connection relations but do not by themselves determine this binary rank.

There is also an exact interpretation in terms of the cut and cycle spaces of the same graph. For any connected graph $G$ on $N$ vertices, let $B$ be its binary incidence matrix, with rows indexed by edges and columns by vertices, and let $\mathcal B=\operatorname{im}B\subseteq\F_2^{E(G)}$ be its cut space. Its orthogonal complement is the cycle space. If every vertex has odd degree, then
\[
 B^{\mathsf T}B=I+A(G),\qquad
 \rank(I+A(G))=N-1-\dim(\mathcal B\cap\mathcal B^\perp).
\]
The first identity follows by counting common incident edges. For the second, $\rank B=N-1$ by connectedness, and the kernel of $B^{\mathsf T}$ restricted to $\operatorname{im}B$ is $\mathcal B\cap\mathcal B^\perp$. The intersection is called the bicycle space, and the second identity is the standard relation with binary Laplacian nullity; see, for example, the proof in~\cite[Proposition~2.11]{Hemmer}. Accordingly, \eqref{eq:rm-open} asks whether its dimension is $N-o(N)$ for $G_m^{\mathrm{sym}}$. Theorem~\ref{thm:bch} already gives this conclusion for the BCH graphs in its stated parameter range.

\section{Conclusions}

The full-parity storage codes obtained from primitive binary BCH codes have asymptotically unit rate for every fixed error budget of at least two. A single bound for polynomial maps proves this statement and also permits an error budget that grows with the extension degree. The bound depends on the Boolean coordinate degree and the dimension of the syndrome space, so it can be applied to other parametrized connection sets with these quantities under control.

For the double-error-correcting family, the intersection of coefficient spaces gives an exact Fibonacci rank formula. The accompanying finite-field lower bound determines the exponential order of the storage deficiency, refining the qualitative unit-rate conclusion in this case.

Determining the largest growth range of the BCH error budget that still gives asymptotically unit rate remains an open direction. For quadratic Reed--Muller evaluation, the rank-one graph on symmetric binary matrices gives an exact formulation of the corresponding question. In that family the target dimension is quadratic in the input dimension, and a stronger rank argument is needed to reach asymptotically unit rate.

\section*{Declarations}
The author received no funding for this work. The author declares no conflict of interest.

\end{document}